\documentclass{article}

\usepackage[T1,T2A]{fontenc}
\usepackage[utf8]{inputenc}
\usepackage[russian,spanish,english]{babel}

\usepackage[a4paper]{geometry}

\usepackage{amsmath}
\usepackage{amsthm}
\usepackage{amsfonts}
\usepackage{amssymb}
\usepackage{graphicx}
\usepackage[colorlinks=true, allcolors=blue]{hyperref}
\usepackage{algorithm}
\usepackage{algorithmic}
\usepackage{tikz}
\usetikzlibrary{quantikz2}
\usepackage{authblk}
\newtheorem{mythm}{Theorem}
\newtheorem{myassum}[mythm]{Assumption}

\newtheorem{myremark}[mythm]{Remark}
\newtheorem{mycoro}[mythm]{Corollary}
\numberwithin{equation}{section}

\usepackage{booktabs}   

\DeclareMathOperator{\E}{\mathbb{E}}

\DeclareMathOperator{\Cov}{Cov}

\DeclareMathOperator{\tr}{tr}

\DeclareMathOperator{\polylog}{polylog}
\newcommand*\dif{\mathop{}\!\mathrm{d}}
\DeclareMathOperator{\Real}{Re}

\title{A quantum gradient descent algorithm for optimizing Gaussian Process models}

\author[1]{Junpeng Hu\thanks{hjp3268@sjtu.edu.cn}}
\author[2]{Jinglai Li\thanks{j.li.10@bham.ac.uk}}
\author[1]{Lei Zhang\thanks{lzhang2012@sjtu.edu.cn}}
\author[1,3]{Shi Jin\thanks{shijin-m@sjtu.edu.cn}}
\affil[1]{School of Mathematical Sciences, Institute of Natural Sciences, MOE-LSC, Shanghai Jiao Tong University, Shanghai, 200240, China}
\affil[2]{School of Mathematics, University of Birmingham, Birmingham, B15 2TT, United Kingdom}
\affil[3]{Shanghai Artificial Intelligence Laboratory, Shanghai, China}

\begin{document}
\maketitle

\begin{abstract}
Gaussian Process Regression (GPR) is a nonparametric  
supervised learning method, widely valued for its ability
to quantify uncertainty. Despite its advantages and broad applications, classical GPR implementations face significant scalability challenges, 
as  they involve matrix operations with a \emph{cubic} complexity in relation to the dataset size.
This computational challenge is further compounded by the  demand of 
optimizing the Gaussian Process model over its hyperparameters,
rendering the total computational cost prohibitive for data intensive problems. 
To address this issue, we propose a quantum gradient descent algorithm 
to optimize the Gaussian Process model. 
Taking advantage of recent advances in quantum algorithms for linear algebra, our algorithm achieves exponential speedup in computing the gradients of the log marginal likelihood. 
The entire gradient descent process is integrated into the quantum circuit. 
Through runtime analysis and error bounds, we demonstrate that our approach significantly improves the scalability of GPR model optimization, making it computationally feasible for large-scale applications.
\end{abstract}


\section{Introduction}
\label{sec:introduction}

Gaussian Process Regression (GPR) is a nonparametric and probabilistic
method for regression and classification~\cite{williams2006gaussian,gramacy2020surrogates}, 
offering robust performance through its principled Bayesian framework and inherent capability to quantify uncertainty. Its applications span a wide range of domains, including astronomy~\cite{foreman2017fast}, geoscience\cite{oliver1990kriging}, psychology~\cite{schulz2018tutorial}, 
material science~\cite{deringer2021gaussian}, 
robotics \cite{deisenroth2013gaussian}, and artificial intelligence (AI)~\cite{shahriari2015taking}, among others. 
A key advantage of GPR is that its kernel-based formulation offers the flexibility to model 
complex functional relationships without presuming a parametric form for the underlying function. However, this flexibility is accompanied by significant computational challenges. 
Namely classical GPR implementations require $\mathcal{O}(N^3)$ operations to compute the inverse of the 
covariance matrix for $N$ training points, making it impractical for large datasets. 
This limitation has historically constrained GPR to relatively small or medium-sized datasets, despite its theoretical advantages over other regression techniques. 
On the other hand, modern real-world applications often involve large-scale, high-dimensional models that demand substantial amounts of data to effectively capture their complexity.
To enable the Gaussian Process (GP) model for large-scale problems, various methods have been developed to reduce this computational burden—such as sparse GP and Blackbox Matrix-Matrix multiplication (BBMM) — but these techniques are inherently approximate and introduce trade-offs in accuracy and scalability. For instance, sparse GP methods~\cite{quinonero2005unifying} reduce computational cost by approximating the full covariance matrix with a subset of inducing points, but their performance heavily depends on the selection and number of these points. Similarly, while BBMM methods can significantly reduce computational costs in many cases, they still rely on certain approximate computations that may lead to inaccurate results~\cite{gardner2018gpytorch}.
Despite these advancements, 
the lack of generally scalable computational methods for GPR limits its applicability to many large-scale problems, necessitating further research on the topic.

Quantum computing is a rapidly advancing computational paradigm with transformative potential: 
 quantum algorithms could offer up to  exponential acceleration over the best-known classical methods~\cite{feynman2018simulating,grover1996fast,shor1994algorithms}. Among these, the quantum algorithm proposed by Harrow, Hassidim, and Lloyd (the HHL algorithm) \cite{harrow2009quantum} for solving linear systems of equations demonstrated an exponential speedup compared to the classical conjugate gradient (CG) method in terms of the matrix size. The original HHL algorithm features a complexity that scales quadratically with the condition number $\kappa$ of the matrix and linearly with the inverse accuracy $1/\varepsilon$. Subsequent works have introduced refinements to mitigate the dependence on both the condition number and the accuracy, enhancing the algorithm's practical applicability \cite{ambainis2012variable,childs2017quantum,costa2022optimal}. In particular, within the framework of block-encoding, the quantum singular value transformation (QSVT) has emerged as a powerful tool for performing matrix arithmetic \cite{gilyen2019quantum}, including efficient matrix inversion. The power of block-encoded matrix powers was studied in \cite{chakraborty2018power}. 
 Since the primary computational bottleneck in GPR arises from matrix operations, particularly matrix inversion, the aforementioned 
 quantum algorithms enable GPR to be implemented on quantum systems with exponential speedup compared to 
 the classical implementation.
 

In this context, considerable efforts have been devoted to developing quantum algorithms for GPR~\cite{zhao2019quantum,das2018quantum,chen2022quantum,farooq2024quantum};
these methods achieve exponential speedup by taking advantage of the quantum algorithms for matrices and linear systems.
Like most machine learning models, the behavior and performance of GPR depend critically on 
its hyperparameters. For all but the simplest problems, the GP model must be optimized with respect to its hyperparameters to ensure its effectiveness. While optimizing the GP model is conceptually straightforward, 
it suffers from the same computational challenge with large datasets: it requires an $N\times N$ matrix inversion in each iteration. 
This issue is considered in \cite{zhao2019quantum}, where the authors propose a quantum algorithm for evaluating the logarithm of the marginal likelihood (LML) of the GP model. In this method, one can only optimize the hyperparameters in a derivative-free manner.
However, gradient-based optimization techniques are generally preferred 
as they are typically more efficient than the derivative-free methods. 
The primary goal of this work 
is to investigate the quantum algorithms for 
optimizing and training GPR, with a particular focus on the gradient-based optimization methods. 
Specifically, we propose a quantum gradient descent method that allows one to optimize the GP model 
using gradient information, 
while preserving the exponential speedup characteristic of quantum computation. A brief comparison between our proposed method and the quantum derivative-free approach by Zhao et al. \cite{zhao2019quantum} is summarized in Table~\ref{tab:comparison}.

\begin{table}[ht]
\centering
\small  
\caption{Comparison between the previous work and the proposed method}
\label{tab:comparison}
\begin{tabular}{l c c c c}
\toprule
\textbf{Method} & 
\textbf{Objective} & 
\textbf{Optimization} & 
\textbf{Convergence} & 
\textbf{Best dependence on $N$} \\
\midrule
Zhao et al. \cite{zhao2019quantum} &
$\log p(\mathbf{y}|\mathbf{X},\boldsymbol{\theta})$ &
Gradient-free &
Linear or worse &
$\polylog(N)$ \\
\midrule
Proposed &
${\frac{\partial}{\partial\theta_{j}}}\log p(\mathbf{y}|\mathbf{X},\boldsymbol{\theta})$ &
Gradient descent &
Linear &
$\polylog(N)$ \\
\bottomrule
\end{tabular}
\end{table}

The structure of the paper is as follows. Section \ref{sec:gp} introduces the GPR framework and discusses the hyperparameter optimization problem.
Section \ref{sec:algorithm} provides a detailed description of the quantum gradient method
for optimizing the GP model. 
Section \ref{sec:analysis} presents a runtime and error analysis of the proposed method. Finally, Section \ref{sec:conclusion}
offers some closing remarks.

\section{GPR and kernel optimization}
\label{sec:gp}

Supervised learning aims to infer the underlying relationship between the input and output of a system using a set of examples, known as the training data. 
GPR provides a nonparametric and probabilistic method for this task. 
Simply put, GPR casts the underlying function as a Gaussian random process, 
and computes the distribution of the function values of interest conditioned on the 
training data in a Bayesian formulation. 
Since GPR obtains a distribution of the function values of interest, it naturally enables uncertainty quantification.
Below, we provide a brief introduction to GPR and refer to \cite{williams2006gaussian} for further details.

\subsection{Overview of GPR}
Given a training set $T= \{ (x_i, y_i)\}^{N-1}_{i=0}$ , where inputs are denoted by $\mathbf{X} := \{x_i\}^{N-1}_{i=0}$, and the corresponding outputs by $\mathbf{y} := \{y_i\}^{N-1}_{i=0}$, we aim to model a latent function $f(x)$ such that
\begin{equation}\label{eqn:gp:noise}
    y = f(x) + \xi,
\end{equation}
where $\xi \sim \mathcal{N}(0, \sigma^2_N )$ represents independent and identically distributed Gaussian observation noise. 
In the GPR framework, the latent function $f(x)$ is assumed to be a Gaussian random process which is fully characterized by its mean function and covariance function:
\begin{equation}
\E[f(x)] = m(x),
\end{equation}
\begin{equation}
    \Cov(f(x),f(x^\prime)) = k(x, x^\prime). 
\end{equation}
For simplicity in practical implementations, the mean function is often assumed to be zero, which can be achieved through appropriate preprocessing of the data. The kernel $k(\cdot,\cdot)$ is positive semidefinite and bounded, and it encapsulates the assumptions about the characteristics of
the target function $f(x)$. 
Now assume that we are interested in the function values at a number of ``test points'' denoted as $\mathbf{X}_*$, and 
let $\mathbf{f}_*= f(\mathbf{X}_*)$. Following the GP assumption, we can derive  
the joint distribution of the observed functions values $\mathbf{y}$ and the function values at the test points $\mathbf{f}_{*}$:   
\begin{equation}\label{eqn:gp:joint}
    \begin{bmatrix}
    \mathbf{y} \\ \mathbf{f}_*
    \end{bmatrix}
    \sim \mathcal{N} \left(\mathbf{0}, 
    \begin{bmatrix}
    K(\mathbf{X}, \mathbf{X}) + \sigma_N^2 \mathbf{I} & K(\mathbf{X}, \mathbf{X}_*) \\
    K(\mathbf{X}_*, \mathbf{X}) & K(\mathbf{X}_*, \mathbf{X}_*)
    \end{bmatrix}
    \right),
\end{equation}
where $K(\mathbf{X}, \mathbf{X})$ denotes the covariance matrix of the training inputs $\mathbf{X}$, $K(\mathbf{X}_{*}, \mathbf{X})$ is the cross-covariance matrix between the test inputs $\mathbf{X}_{*}$ and $\mathbf{X}$, and $K(\mathbf{X}_{*}, \mathbf{X}_{*})$ is the covariance matrix of the test inputs. The term $\sigma_N^2 \mathbf{I}$ accounts for the Gaussian noise in the observed target values.
From the joint distribution~\eqref{eqn:gp:joint}, we can derive the posterior distribution of the function values of interest $\mathbf{f}_*$ conditioned on the
observation data, which is also Gaussian: 
\begin{equation}
    \mathbf{f}_* \mid \mathbf{X}, \mathbf{y}, \mathbf{X}_* \sim \mathcal{N}(\boldsymbol{\mu}_*, \Sigma_*),
\end{equation}
with the posterior mean $\boldsymbol{\mu}_*$ and posterior covariance $\Sigma_*$ given by:  
\begin{equation}
\begin{aligned}
    \boldsymbol{\mu}_* &= K(\mathbf{X}_*, \mathbf{X}) \big[K(\mathbf{X}, \mathbf{X}) + \sigma_N^2 \mathbf{I}\big]^{-1} \mathbf{y}, \\
    \Sigma_* &= K(\mathbf{X}_*, \mathbf{X}_*) - K(\mathbf{X}_*, \mathbf{X}) \big[K(\mathbf{X}, \mathbf{X}) + \sigma_N^2 \mathbf{I}\big]^{-1} K(\mathbf{X}, \mathbf{X}_*).
\end{aligned}
\end{equation}
The posterior mean $\boldsymbol{\mu}_*$ provides an estimate of the function values at the test locations, while the posterior covariance $\Sigma_*$ quantifies the uncertainty in these predictions. 
We emphasize that both $\boldsymbol{\mu}_*$ and $\Sigma_*$ depends on $\big[K(\mathbf{X}, \mathbf{X}) + \sigma_N^2 \mathbf{I}\big]^{-1}$, 
and as such their computation requires $\mathcal{O}(N^3)$ operations in the classical setting.

\subsection{Kernel optimization}

As has been mentioned earlier, the kernel functions in GPR usually involve hyperparameters that are crucial for the performance of the method.
Consider the widely used Matern kernel as an example: 
\[
k(x, x') = \sigma^2 \frac{2^{1-\nu}}{\Gamma(\nu)} \left( \sqrt{2\nu} \frac{\|x - x'\|}{\ell} \right)^\nu K_\nu \left( \sqrt{2\nu} \frac{\|x - x'\|}{\ell} \right),
\]
where $ K_\nu $ is the modified Bessel function of the second kind, $ \Gamma(\nu) $ is the gamma function.
The kernel has three hyperparameters: the variance $ \sigma^2 $,
the length-scale $ \ell >0 $, and  $ \nu > 0 $ that controls the smoothness of the target function.
All of these hyperparameters play an important role in determining the flexibility, expressiveness, and overall performance of the GP model. 
As such an essential step in GPR is to optimize the covariance kernel with respect to these hyperparameters,
and this is often done by a maximum likelihood estimation approach to ensure the kernel accurately models the underlying data.
We assume a general covariance function $k_{\boldsymbol{\theta}}$, parameterized by a vector of hyperparameters $\boldsymbol{\theta}$. 
We shall find the optimal values for $\boldsymbol{\theta}$ by maximizing the marginal likelihood $p(\mathbf{y} \mid \mathbf{X}, \boldsymbol{\theta})$ --
it is called ``marginal'' likelihood because 
it is considered to be obtained by marginalizing the joint distribution~\eqref{eqn:gp:joint} over $\mathbf{f}_*$. 
It can be found that the logarithm of the marginal likelihood (LML), a more convenient form for optimization, is given by:  
\begin{equation}\label{eqn:gp:LML}
\log p(\mathbf{y}|\mathbf{X},\boldsymbol{\theta})=-\frac{1}{2}\mathbf{y}^{\top}K^{-1}\mathbf{y}-\frac{1}{2}\log|K|-\frac{N}{2}\log2\pi,
\end{equation}
where $K = K_f + \sigma^2_n \mathbf{I}$ is the covariance matrix of the noisy observations $\mathbf{y}$, and $K_f=K(\mathbf{X},\mathbf{X})$.

By explicitly evaluating the LML for various hyperparameter settings, one can employ gradient-free optimization techniques—such as grid search, random search, or Bayesian optimization—to identify the optimal kernel parameters $\boldsymbol{\theta}$. 
A key computational challenge is that, as can be seen in Eq.~\eqref{eqn:gp:LML}, evaluating LML also has a complexity of $\mathcal{O}(N^3)$.  
In \cite{zhao2019quantum}, the authors developed a quantum algorithm to address this issue. 
It is shown in their work that, under certain conditions, their algorithm reduces the classical computational complexity from $\mathcal{O}(N^3)$ to logarithmic scaling with respect to $N$, 
making the computation feasible for large datasets.

As is discussed in Section~\ref{sec:introduction},  gradient based methods are often more efficient in many applications.
In such methods, we need to compute not only the function value of LML (which is discussed in \cite{zhao2019quantum}) but also the gradient of it. 
The gradient of LML can be formally derived as, 
\begin{equation}\label{eqn:gp:dLML}
\begin{aligned}
    {\frac{\partial}{\partial\theta_{j}}}\log p(\mathbf{y}|X,\boldsymbol{\theta})& ={\frac{1}{2}}\mathbf{y}^{\top}K^{-1}{\frac{\partial K}{\partial\theta_{j}}}K^{-1}\mathbf{y}-{\frac{1}{2}}\tr\left(K^{-1}{\frac{\partial K}{\partial\theta_{j}}}\right),
\end{aligned}
\end{equation}
for $j=1,...d$ (assuming $\boldsymbol{\theta}$ is $d$-dimensional). 
One can see from Eq~\eqref{eqn:gp:dLML},
the computational complexity for evaluating the LML derivatives is the same (i.e. $\mathcal{O}(N^3)$) as evaluating LML itself (Eq.~\eqref{eqn:gp:LML}),
which is primarily determined by inverting the $N\times N$ matrix $K$. 
In the subsequent section, we provide a quantum gradient descent algorithm that can achieve exponential 
speadup compared to the classical implementations. 

\section{Quantum algorithm}
\label{sec:algorithm}

Without loss of generality, let us consider the optimization of a single hyperparameter $\theta$ in the feasible parameter space, with learning rate $\eta \in \mathbb{R}^+$ satisfying $\eta \ll 1$ to ensure convergence of the gradient descent procedure. The update of multiple hyperparameters, $\boldsymbol{\theta}$, is performed naturally by computing the partial derivatives with respect to each individual hyperparameter. Moreover, a randomized coordinate descent method can be employed to further reduce computational complexity \cite{ding2024random}. We begin by presenting the quantum circuit construction for evaluating the real component of the inner product $\bra{\varphi_1} U \ket{\varphi_2}$ in Section \ref{sec:algorithm:aUb}, where $U$ represents a unitary operator and $\ket{\varphi_1}, \ket{\varphi_2}$ denote arbitrary quantum states in the corresponding Hilbert space. Gradients are then expressed in this form in Section \ref{sec:algorithm:derivative}.
In Section \ref{sec:algorithm:iteration} we provide a complete quantum gradient descent algorithm specifically designed for optimizing the GP model.

\subsection{Estimate the inner product}
\label{sec:algorithm:aUb}
Let $\ket{\varphi_1}, \ket{\varphi_2} \in \mathcal{H} = \mathbb{C}^{2^q}$ be quantum states that can be prepared by unitary operators $U_{\varphi_1}, U_{\varphi_2} \in \mathcal{U}(\mathcal{H})$ respectively, such that $\ket{\varphi_i} = U_{\varphi_i}\ket{0}$ for $i \in {1,2}$. The quantum circuit implementation of the Hadamard test, as illustrated in Figure \ref{fig:circuit:hadamard},
\begin{figure}[htbp]
    \centering
    \begin{quantikz}
        \lstick{$\ket{0}$} & \gate[2]{U_{\psi_0}} &  \\
        \lstick{$\ket{0^q}$} &   &   \\
    \end{quantikz} 
    \quad := \quad 
    \begin{quantikz}
        \lstick{$\ket{0}$} & \gate{H} & \gate{X} & \ctrl{1} & \gate{X} & \ctrl{1} & \ctrl{1} & \gate{H} & \\
        \lstick{$\ket{0^q}$} &  &  & \gate{U_{\varphi_1}} &  & \gate{U_{\varphi_2}} & \gate{U} &  & \\
    \end{quantikz}
    \caption{Quantum circuit for the Hadamard test}.
    \label{fig:circuit:hadamard}
\end{figure}
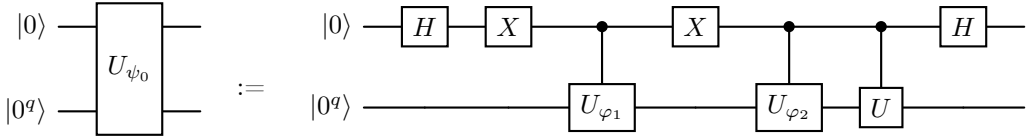
maps $\ket{0}\ket{0^q}$ to $\ket{\psi_0}$:
  \begin{equation}
  \begin{aligned}
      \ket{\psi_0} = U_{\psi_0} \ket{0}\ket{0^q} &= \frac{1}{2} \ket{0} \left( \ket{\varphi_1} + U\ket{\varphi_2} \right) + \frac{1}{2} \ket{1} \left( \ket{\varphi_1} - U\ket{\varphi_2} \right) \\
      &\triangleq \sin{\frac{\alpha}{2}} \ket{\text{good}} + \cos{\frac{\alpha}{2}} \ket{\text{bad}},
  \end{aligned}
  \end{equation}
where $\sin^2 \frac{\alpha}{2} = \left\| \ket{\varphi_1} + U\ket{\varphi_2} \right\|^2/4 = (1+\Real\bra{\varphi_1} U \ket{\varphi_2}) / 2 $. By measuring the probability of the ancilla qubit being in the state $ \ket{0} $, $ \sin^2 \frac{\alpha}{2} $ can be estimated, enabling the calculation of the inner product. However, this approach requires $\mathcal{O}(1/\varepsilon^2)$ measurements, which also introduces a statistical error.

The principal methodology leverages amplitude estimation to enhance the efficiency of the Hadamard test \cite{lin2022lecture}. Let us define an orthonormal basis $\mathcal{B} = {\ket{\text{good}}, \ket{\text{bad}}}$ and the corresponding two-dimensional Hilbert space $\mathcal{H}_{0} = \text{span} \ \mathcal{B}$. Define the Grover operator $G \in \mathcal{U}(\mathbb{C}^{2} \otimes \mathcal{H})$ as the composition of two reflection operators $G = R_{\psi_0} R_{\text{good}}$, where
\begin{equation}
\begin{aligned}
    R_{\psi_0} &:= 2\ket{\psi_0} \bra{\psi_0} - I_{q+1} = U_{\psi_0} (2 \ket{0^{q+1}}\bra{0^{q+1}} - I_{q+1}) U_{\psi_0}^\dagger, \\
    R_{\text{good}} &:= I_{q+1} - 2 \ket{\text{good}} \bra{\text{good}}  \overset{\mathcal{H}_{0}}{=} (I-2\ket{0}\bra{0}) \otimes I_q.
\end{aligned}
\end{equation}
The aforementioned equality holds within the subspace $\mathcal{H}_{0}$. In this representation, the Grover operator $G$ admits a quantum circuit implementation as depicted in Figure \ref{fig:circuit:G}.

\begin{figure}[htbp]
    \centering
    \begin{quantikz}
        \lstick[2]{$\ket{\psi}$} & & \gate{X}\gategroup[2,steps=3,style={inner sep=2pt}]{$R_{\text{good}}$} & \gate{Z} & \gate{X} & \gate[2]{U_{\psi_0^\dagger}}\gategroup[2,steps=5,style={inner sep=2pt}]{$R_{\psi_0}$} & \gate{X} & \gate{Z} & \gate{X} & \gate[2]{U_{\psi_0}} & \\
        & \qwbundle{q} &  &  &  &  & \gate{X^{\otimes q}} & \ctrl{-1} & \gate{X^{\otimes q}} &  & 
    \end{quantikz}
    \caption{Quantum circuit for the Grover operator $G$}.
    \label{fig:circuit:G}
\end{figure}
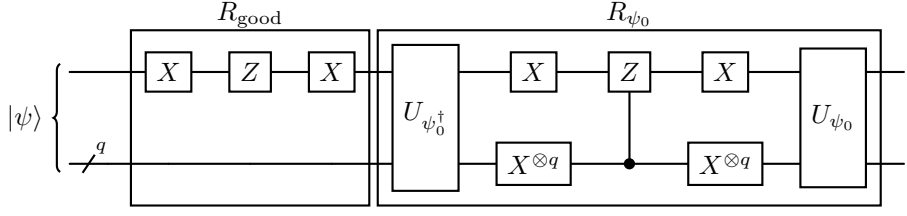

The subspace $\mathcal{H}_{0}$ is an invariant subspace of the operator $G$, represented by the matrix
\begin{equation}
    G|_{\mathcal{B}} = \begin{pmatrix}
        \cos \alpha & \sin \alpha \\ -\sin \alpha & \cos \alpha
    \end{pmatrix}.
\end{equation}
Its two eigenstates are given by
\begin{equation}
    \ket{\psi_\pm} = \frac{\ket{\text{good}} \pm i \ket{\text{bad}}}{\sqrt{2}},
\end{equation}
with eigenvalues $e^{\pm i \alpha}$ respectively. From this point forward, we assume that the $m$-bit fixed point representation ($0.\alpha_{m-1}\cdots\alpha_{0})$ of the phase $\alpha/2\pi$ is exact. 

We now present the application of quantum phase estimation (QPE), originally introduced by Kitaev \cite{kitaev1995quantum}, for amplitude estimation. For a comprehensive theoretical treatment, we refer to \cite{lin2022lecture, rieffel2011quantum}.
We can define a controlled unitary operation as follows:
\begin{equation}
    \mathcal{U} = \sum_{j\in[2^m]} \ket{j}\bra{j} \otimes G^j = \sideset{}{'}\prod_{j = 0}^{m-1} \left( \ket{0}\bra{0} \otimes I + \ket{1} \bra{1} \otimes G^{2^j} \right),
\end{equation}
The primed product $\sideset{}{'}\prod$ denotes a hybrid operation: it represents the tensor product for the first register and the standard matrix product for the second register. Utilizing the Quantum Fourier Transform (QFT) and the operator $\mathcal{U}$, we transform the initial state $\ket{0^m}\ket{\psi_0}$ into
\begin{equation}\label{eqn:QPE}
    \begin{aligned}
        \ket{0^m} \ket{\psi_0} \xrightarrow{U_{\text{QFT}}\otimes I} & \frac{1}{\sqrt{2^m}} \sum_{j\in[2^m]} \ket{j} \ket{\psi_0} \\
        \xrightarrow{\mathcal{U}} & \frac{1}{\sqrt{2^m}} \sum_{j\in[2^m]} \ket{j} G^{j} \ket{\psi_0} = \frac{1}{\sqrt{2^m}} \sum_{j\in[2^m]} \ket{j} G^{j} \left( \frac{e^{i\frac{\alpha}{2}}}{\sqrt{2}} \ket{\psi_+} + \frac{e^{-i\frac{\alpha}{2}}}{\sqrt{2}} \ket{\psi_{-}} \right) \\
        \xrightarrow{U_{\text{QFT}}^\dagger \otimes I} & \frac{1}{\sqrt{2}} e^{i\frac{\alpha}{2}} \ket{\frac{\alpha}{2\pi}} \ket{\psi_{+}} + \frac{1}{\sqrt{2}} e^{-i\frac{\alpha}{2}} \ket{\frac{2\pi-\alpha}{2\pi}} \ket{\psi_{+}},
    \end{aligned}
\end{equation}
which corresponds precisely to the quantum phase estimation process, denoted as $U_{\text{QPE}}$ hereafter. The quantum circuit for $U_{\text{QPE}}$ is illustrated in Figure \ref{fig:circuit:QPE}.

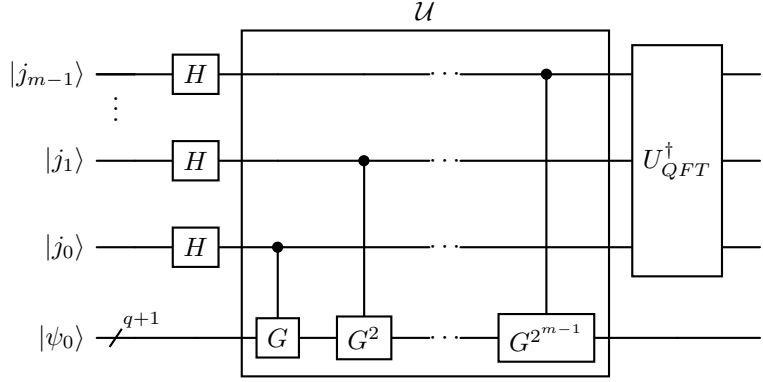
\begin{figure}[htbp]
    \centering
    \begin{quantikz}
        \lstick{$\ket{j_{m-1}}$} & \qw{\vdots} & \gate{H} & \gategroup[4,steps=4,style={inner sep=2pt}]{$\mathcal{U}$} &  & \cdots & \ctrl{3} & \gate[3]{U_{QFT}^\dagger} & \\
        \lstick{$\ket{j_{1}}$} &  & \gate{H} &  & \ctrl{2} & \cdots &  &  & \\
        \lstick{$\ket{j_{0}}$} &  & \gate{H} & \ctrl{1} &  & \cdots &  &  & \\
        \lstick{$\ket{\psi_0}$} & \qwbundle{q+1} &  & \gate{G} & \gate{G^2} & \cdots & \gate{G^{2^{m-1}}} &  & 
    \end{quantikz}
    \caption{Quantum circuit for the quantum phase estimation $U_{\text{QPE}}$}.
    \label{fig:circuit:QPE}
\end{figure}

Utilizing  \textit{classical arithmetics circuits} \cite{rieffel2011quantum, haner2018optimizing}, we can construct a mapping $U_{\cos}$ such that
\begin{equation}
    U_{\cos} \ket{0^m} \ket{\lambda} = \ket{-\cos(2\pi\lambda)} \ket{\lambda}.
\end{equation}
By applying $U_{\cos}$ to $m$ ancilla qubits and the first register in Equation \eqref{eqn:QPE}, we obtain
\begin{equation}
    U_{\cos} \ket{0^m} \ket{\frac{\alpha}{2\pi}} = \ket{-\cos \alpha} \ket{\frac{\alpha}{2\pi}}, \quad 
    U_{\cos} \ket{0^m} \ket{\frac{2\pi - \alpha}{2\pi}} = \ket{-\cos \alpha } \ket{\frac{2\pi - \alpha}{2\pi}},
\end{equation}
with the fact
\begin{equation}
    \Real\bra{\varphi_1}U\ket{\varphi_2} = 2 \sin^2 \frac{\alpha}{2} - 1 = - \cos \alpha.
\end{equation}
Notably, the information of $\Real\bra{\varphi_1}U\ket{\varphi_2}$ is encoded in the ancilla qubits. Finally, after performing the uncomputation by applying $U_{\text{QPE}}^\dagger$ to the state described in Equation \eqref{eqn:QPE}, we arrive at
\begin{equation}
    \ket{\psi_{*}} = \ket{\Real\bra{\varphi_1}U\ket{\varphi_2}} \ket{0^m} \ket{\psi_0}.
\end{equation}
It can be further simplified as
\begin{equation}
\begin{aligned}
    \ket{\psi_{*}} \xrightarrow{I_{m}\otimes I_{m}\otimes U_{\psi_0}^\dagger} & \ket{\Real\bra{\varphi_1}U\ket{\varphi_2}} \ket{0^m} \ket{0^{q+1}} \\
\end{aligned}    
\end{equation}
The entire procedure is illustrated in Figure \ref{fig:circuit:aUb}.

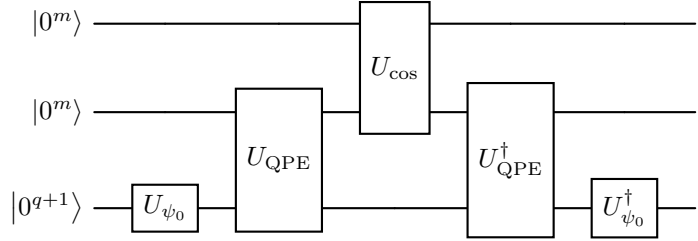
\begin{figure}[htbp]
    \centering
    \begin{quantikz}
        \lstick{$\ket{0^m}$} &  &  & \gate[2]{U_{\cos}} &  &  & \\
        \lstick{$\ket{0^m}$} &  & \gate[2]{U_{\text{QPE}}} &  & \gate[2]{U_{\text{QPE}}^\dagger} &  & \\
        \lstick{$\ket{0^{q+1}}$} & \gate{U_{\psi_0}} &  &  &  & \gate{U_{\psi_0}^\dagger} &
    \end{quantikz}
    \caption{Quantum circuit for computing $\Real\bra{\varphi_1}U\ket{\varphi_2}$}.
    \label{fig:circuit:aUb}
\end{figure}

\subsection{Derivative of the marginal likelihood}
\label{sec:algorithm:derivative}
Now we consider  the computation of the LML derivative. 
In particular we will show how to transform the derivative into the form $\bra{\varphi_1} U \ket{\varphi_2}$ and provide a detailed construction of the operator $U$. First, we outline the necessary assumptions regarding the covariance matrix $K$.

\begin{myassum}[Bounded Covariance Matrix]
    The covariance matrix is uniformly bounded and the derivative of it is uniformly upper bounded, i.e.,
    \begin{equation}
        0 < C_0 \leq \left\| K \right\| \leq C_1, \quad \left\| \frac{\dif K}{\dif \theta} \right\| \leq C_1,
    \end{equation}
    where $C_0$, $C_1$ are independent of $\theta$.
\end{myassum}
For simplicity we also assume that $C_1 = 1$. Let $\kappa$ denote the condition number of $K$. The singular values of $K$ are thus contained in the interval $[1/\kappa,1]$, with $\sigma_N^2 \leq 1/\kappa$ since $K = K_f + \sigma_N^2 I$. Moreover, we can define $C_0 = \sigma_N^2 / \beta$ with $\beta>1$. For instance, we can take $\beta = 4/3$.

\begin{myassum}[Oracles]\label{assump:oracles}
    We are provided with the oracles $O_{\mathbf{y}}$, $O_{K}$ and $O_{\dif K}$ that satisfy the following conditions:
    \begin{equation*}
    \begin{aligned}
      &  O_{\mathbf{y}} \ket{0^n} = \ket{\mathbf{y}}, \\
      &  O_{K} \ket{\mathbf{x}} \ket{0^a} \ket{\theta} = \left(K(\theta)\ket{\mathbf{x}} \ket{0^a} + \ket{\perp}\right) \ket{\theta}, \\
      &  O_{\dif K} \ket{\mathbf{x}} \ket{0^b} \ket{\theta} = \left( \frac{\dif K}{\dif \theta}(\theta)\ket{\mathbf{x}} \ket{0^b} + \ket{\perp}\right) \ket{\theta},
    \end{aligned}
    \end{equation*}
    where 
   \begin{itemize}
   \item $\ket{\theta} = \ket{\theta_{0} \cdots \theta_{m-1}}$ is the binary representation of $\theta$, 
    \item $\ket{\mathbf{y}} = \sum_{i=0}^{N-1} y_{i} \ket{i}$ and  $N=2^n$,
    \item $\ket{\perp}$ denotes some orthogonal state (not necessarily normalized) such that $(I \otimes \ket{0^c}\bra{0^c}) \ket{\perp} = 0$ with $\ket{0^c}, c=a,b$ the ancilla register.
    \end{itemize}
\end{myassum}

Given the block encoding $O_{K}$ of the covariance matrix $K$, several quantum linear system solvers can be employed to obtain a block encoding of $K^{-1}$, such as QSVT and HHL. For further details of these methods, we refer readers to \cite{harrow2009quantum, gilyen2019quantum,chakraborty2018power}. In this context, we define the procedure as an oracle $O_{K^{-1}}$:
\begin{equation}
    O_{K^{-1}} \ket{\mathbf{x}} \ket{0^{a+1}} \ket{\theta} = \left( C_0 K(\theta)^{-1}\ket{\mathbf{x}} \ket{0^{a+1}} + \ket{\perp}\right) \ket{\theta}.
\end{equation}
To calculate the derivative of LML, we first reformulate it into the form of $\bra{\varphi_1} U \ket{\varphi_2}$,
\begin{equation}\label{eqn:derivative}
\begin{aligned}
    {\frac{\partial}{\partial\theta_{j}}}\log p(\mathbf{y}|X,\mathbf{\theta})&= {\frac{1}{2}}\mathbf{y}^{\top}K^{-1}{\frac{\partial K}{\partial\theta_{j}}}K^{-1}\mathbf{y}-{\frac{1}{2}}\tr\left(K^{-1}{\frac{\partial K}{\partial\theta_{j}}}\right) \\
    &= {\frac{1}{2}}\mathbf{y}^{\top}K^{-1}{\frac{\partial K}{\partial\theta_{j}}}K^{-1}\mathbf{y}-{\frac{1}{2}} \sum_{j\in[2^n]} \mathbf{e}_{j}^{\top} K^{-1}{\frac{\partial K}{\partial\theta_{j}}} \mathbf{e}_{j}
\end{aligned}
\end{equation}
We combine all the column vectors into large column vectors
\begin{equation*}
\begin{aligned}
    \tilde{\mathbf{y}}_{l} &= \left[\mathbf{e}_0;\cdots;\mathbf{e}_{2^n-1};\mathbf{y};\mathbf{0};\cdots;\mathbf{0}\right] = \ket{0} \otimes \left(\sum_{j\in[2^n]} \ket{j} \otimes \mathbf{e}_{j} \right) + \ket{1} \otimes \ket{0^n} \otimes \mathbf{y}, \\
    \tilde{\mathbf{y}}_{r} &= \left[-C_0\mathbf{e}_0;\cdots;-C_0\mathbf{e}_{2^n-1};\mathbf{y};\mathbf{0};\cdots;\mathbf{0}\right] = \ket{0} \otimes \left(\sum_{j\in[2^n]} \ket{j} \otimes -C_0 \mathbf{e}_{j} \right) + \ket{1} \otimes \ket{0^n} \otimes \mathbf{y}.
\end{aligned}
\end{equation*}
The quantum states are defined as
\begin{equation*}
\begin{aligned}
    \ket{\tilde{\mathbf{y}}_{l}} &= \frac{1}{\mathcal{N}_{l}} \tilde{\mathbf{y}}_{l} =  \sum_{j\in[2^n]} \frac{1}{\mathcal{N}_{l}} \ket{0}_{i_1} \ket{j}_{i_2} \ket{j}_{w} + \frac{\left\| \mathbf{y} \right\|}{\mathcal{N}_{l}} \ket{1}_{i_1}\ket{0^n}_{i_2}\ket{\mathbf{y}}_{w}, \\
    \ket{\tilde{\mathbf{y}}_{r}} &= \frac{1}{\mathcal{N}_{r}} \tilde{\mathbf{y}}_{r} =  \sum_{j\in[2^n]} \frac{-C_0}{\mathcal{N}_{r}} \ket{0}_{i_1} \ket{j}_{i_2} \ket{j}_{w} + \frac{\left\| \mathbf{y} \right\|}{\mathcal{N}_{r}} \ket{1}_{i_1}\ket{0^n}_{i_2}\ket{\mathbf{y}}_{w},
\end{aligned}
\end{equation*}
where the subscripts $i_1$, $i_2$ indicate `index register', $w$ indicates `work register', and $\mathcal{N}_{l}$, $\mathcal{N}_{r}$ are the normalization factors
\begin{equation*}
    \mathcal{N}_{l} = \left\| \tilde{\mathbf{y}}_{l} \right\| = \sqrt{ \sum_{j=0}^{2^n-1} \left\| \mathbf{e}_{j} \right\|^2 + \left\| \mathbf{y} \right\|^2} = \sqrt{2^n + \left\| \mathbf{y} \right\|^2}, \quad \mathcal{N}_{r} = \sqrt{2^nC_0^2 + \left\| \mathbf{y} \right\|^2}.
\end{equation*}
Given the oracle $O_{\mathbf{y}}$, we can construct the unitaries $U_{\tilde{\mathbf{y}}_{l}}$ and $U_{\tilde{\mathbf{y}}_{r}}$ to prepare $\ket{\tilde{\mathbf{y}}_{l}}$ and $\ket{\tilde{\mathbf{y}}_{r}}$ respectively. For detailed analysis, see Appendix \ref{sec:appendix:initial}. The initial states are prepared as
\begin{equation*}
    \ket{\varphi_l} = \ket{\tilde{\mathbf{y}}_{l}} \ket{0^{a+1}}_{a_1} \ket{0^{a+1}}_{a_2} \ket{0^{b}}_{b} \ket{\theta}, \quad \ket{\varphi_r} = \ket{\tilde{\mathbf{y}}_{r}} \ket{0^{a+1}}_{a_1} \ket{0^{a+1}}_{a_2} \ket{0^{b}}_{b} \ket{\theta},
\end{equation*}
where $a_1$, $a_2$ and $b$ indicate `ancilla register' used for the oracles.

Next, we construct a unitary $U$ so that the derivative of the covariance function \eqref{eqn:derivative} can be computed by $\bra{\varphi_l} U \ket{\varphi_r}$. Let 
\begin{equation}
    U = \left[ O_{K^{-1}} \right]_{w,a_2,\theta}  \left[ O_{dK} \right]_{w,b,\theta} \left[ O_{K^{-1}} \right]_{w, a_1,\theta}^{i_1}, \label{eq:U}
\end{equation}
where the subscript denotes the target register and the superscript denotes the control register. Applying $U$ to the initial state $\ket{\varphi_{r}}$, we obtain
\begin{equation*}
\begin{aligned}
    &\ket{\tilde{\mathbf{y}}_{r}} \ket{0^{a+1}}_{a_1} \ket{0^{a+1}}_{a_2} \ket{0^{b}}_{b} \ket{\theta} \\
    \xrightarrow{\left[O_{K^{-1}}\right]_{w, a_1, \theta}^{i_1}} &  \sum_{j\in[2^n]} \frac{-C_0}{\mathcal{N}_{r}} \ket{0}_{i_1} \ket{j}_{i_2} \ket{j}_{w} \ket{0^{a+1}}_{a_1} \ket{0^{a+1}}_{a_2} \ket{0^{b}}_{b} \ket{\theta} \\
    &+ \frac{\left\| \mathbf{y} \right\|}{\mathcal{N}_{r}} \ket{1}_{i_1}\ket{0^n}_{i_2} C_0 K(\theta)^{-1}\ket{\mathbf{y}}_{w} \ket{0^{a+1}}_{a_1} \ket{0^{a+1}}_{a_2} \ket{0^{b}}_{b} \ket{\theta} + \ket{\perp_{a_1}}\\
    \xrightarrow{\left[ O_{dK} \right]_{w,b,\theta}} & \sum_{j\in[2^n]} \frac{-C_0}{\mathcal{N}_{r}} \ket{0}_{i_1} \ket{j}_{i_2} \frac{\dif K}{\dif \theta} (\theta) \ket{j}_{w} \ket{0^{a+1}}_{a_1} \ket{0^{a+1}}_{a_2} \ket{0^{b}}_{b} \ket{\theta} \\
    &+ \frac{C_0\left\| \mathbf{y} \right\|}{\mathcal{N}_{r}} \ket{1}_{i_1}\ket{0^n}_{i_2} \frac{\dif K}{\dif \theta}(\theta) K(\theta)^{-1}\ket{\mathbf{y}}_{w} \ket{0^{a+1}}_{a_1} \ket{0^{a+1}}_{a_2} \ket{0^{b}}_{b} \ket{\theta} + \ket{\perp_{a_1}} + \ket{\perp_{b}} \\
    \xrightarrow{\left[ O_{K^{-1}} \right]_{w,a_2,\theta}} & \sum_{j\in[2^n]} \frac{-C_0^2}{\mathcal{N}_{r}} \ket{0}_{i_1} \ket{j}_{i_2} K(\theta)^{-1} \frac{\dif K}{\dif \theta} (\theta) \ket{j}_{w} \ket{0^{a+1}}_{a_1} \ket{0^{a+1}}_{a_2} \ket{0^{b}}_{b} \ket{\theta} \\
    &+ \frac{C_0^2\left\| \mathbf{y} \right\|}{\mathcal{N}_{r}} \ket{1}_{i_1}\ket{0^n}_{i_2} K(\theta)^{-1} \frac{\dif K}{\dif \theta}(\theta) K(\theta)^{-1}\ket{\mathbf{y}}_{w} \ket{0^{a+1}}_{a_1} \ket{0^{a+1}}_{a_2} \ket{0^{b}}_{b} \ket{\theta} \\
    &+ \ket{\perp_{a_1}} + \ket{\perp_{b}} + \ket{\perp_{a_2}},
\end{aligned}
\end{equation*}
where $\ket{\perp_{m}}$ is orthogonal to the ancilla register indexed $m$. The quantum circuit for $U$ is illustrated in Figure \ref{fig:circuit:U}. 

\begin{figure}[htbp]
    \centering
    \begin{quantikz}
        \lstick{$i_1$} &  &  & \ctrl{2} &  &  &  \\
        \lstick{$i_2$} & \qwbundle{n} &  &  &  &  & \\
        \lstick{$w$} & \qwbundle{n} &  &  \gate[5][1.7cm]{O_{K^{-1}}}\gateinput{$\bullet$} &  \gate[5][1.7cm]{O_{dK}}\gateinput{$\bullet$}  & \gate[5][1.7cm]{O_{K^{-1}}}\gateinput{$\bullet$} &  \\
        \lstick{$a_1$} & \qwbundle{a+1} &  & \gateinput{$\bullet$} &  &  &  \\
        \lstick{$a_2$} & \qwbundle{a+1} &  &  &  & \gateinput{$\bullet$} &  \\
        \lstick{$b$} & \qwbundle{b} &  &  & \gateinput{$\bullet$} &  &  \\
        \lstick{$\theta$} & \qwbundle{m} &  & \gateinput{$\bullet$} & \gateinput{$\bullet$} & \gateinput{$\bullet$} &  \\
    \end{quantikz}
    \caption{Quantum circuit for the unitary operator $U$. The bullet points indicate the target qubits to which the gate is applied.}.
    \label{fig:circuit:U}
\end{figure}
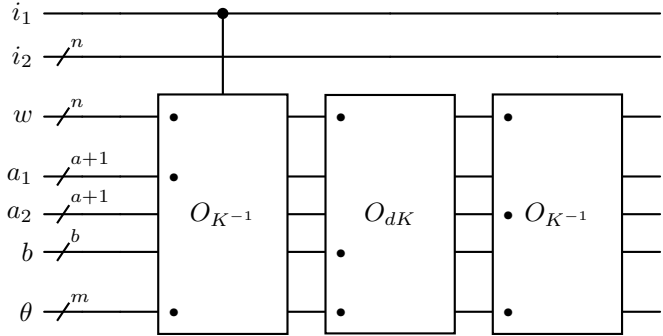

Finally it can be verified that
\begin{equation*}
\begin{aligned}
    &\quad \bra{\varphi_l} U \ket{\varphi_r} \\
    &= \bra{\tilde{\mathbf{y}}_{l}} \left( \sum_{j\in[2^n]} \frac{-C_0^2}{\mathcal{N}_{r}} \ket{0}_{i_1} \ket{j}_{i_2} K(\theta)^{-1} \frac{\dif K}{\dif \theta} (\theta) \ket{j}_{w} + \frac{C_0^2\left\| \mathbf{y} \right\|}{\mathcal{N}_{r}} \ket{1}_{i_1}\ket{0^n}_{i_2} K(\theta)^{-1} \frac{\dif K}{\dif \theta}(\theta) K(\theta)^{-1}\ket{\mathbf{y}}_{w} \right) \\
    &= \frac{C_0^2}{\mathcal{N}_{l} \mathcal{N}_{r}} \left( -\sum_{j\in[2^n]} \bra{j}_{w} K(\theta)^{-1} \frac{\dif K}{\dif \theta} (\theta) \ket{j}_{w} + \left\| \mathbf{y} \right\|^2 \bra{\mathbf{y}}_{w} K(\theta)^{-1} \frac{\dif K}{\dif \theta}(\theta) K(\theta)^{-1}\ket{\mathbf{y}}_{w} \right) \\
    &= \frac{2 C_0^2}{\mathcal{N}_{l} \mathcal{N}_{r}} \frac{\partial}{\partial\theta}\log p(\mathbf{y}|X,\mathbf{\theta}),
\end{aligned}
\end{equation*}
which indicates that the LML derivative can be computed via $\bra{\varphi_l} U \ket{\varphi_r}$
with $U$ given by Eq.~\eqref{eq:U}.

\subsection{Iteration} \label{sec:algorithm:iteration}
In Section \ref{sec:algorithm:derivative} we discussed the quantum computation of the gradient of LML.
Now we consider the complete gradient descent algorithm for optimizing the LML function. 
In a standard gradient descent iteration, the design parameter $\theta$ is updated as, 
\begin{equation}\label{eqn:GD}
\begin{aligned}
    \theta_{t+1} &= \theta_{t} - \eta_t \frac{\partial}{\partial\theta}\log p(\mathbf{y}|X,\mathbf{\theta}) \\
    &= \theta_{t} - \frac{\eta_t \mathcal{N}_{l}\mathcal{N}_{r}}{2C_0^2} \bra{\varphi_l} U \ket{\varphi_r} \\
    &\triangleq \theta_{t} - \frac{\mu_t}{2C_0^2} \bra{\varphi_l} U \ket{\varphi_r},
\end{aligned}
\end{equation}
where $\eta_t$ is the step size and $\mu_t = \eta_t \mathcal{N}_{l} \mathcal{N}_{r}$.
Since $\theta = \mathcal{O}(1)$, $\eta_t$ should be chosen small enough to satisfy $\mu_t = \mathcal{O}(1)$.

\begin{algorithm}
  \caption{Hybrid quantum-classical gradient descent algorithm for optimizing GP}\label{alg:hybrid}
  \begin{algorithmic}[1]
  \REQUIRE Oracles $O_{\mathbf{y}}$, $O_{K}$, $O_{dK}$. The initial parameter $\theta_{0}$. 
  \STATE Construct $U_{\tilde{\mathbf{y}}_{l}}$, $U_{\tilde{\mathbf{y}}_{r}}$ as shown in Figure \ref{fig:circuit:initial};
  \STATE\label{alg:step:K-1} Construct $O_{K^{-1}}$ using QSVT and then $U$ as shown in Figure \ref{fig:circuit:U};
  \STATE Construct $U_{\psi_{0}}$ and then $G$ as shown in Figure \ref{fig:circuit:hadamard}, Figure \ref{fig:circuit:G};
  \FOR{$t = 0, 1, \dots, T-1$}
  \STATE\label{alg:step:qpe} Estimate $\bra{\varphi_l} U \ket{\varphi_r}$ by measuring the ancilla qubit in the Hadmard test;
  \STATE Update $\theta$ classically as $\theta_{t+1} \leftarrow \theta_{t} - \frac{\mu_t}{2C_0^2} \bra{\varphi_l} U \ket{\varphi_r}$;
  \STATE Estimate the LML as shown in \cite{zhao2019quantum};
  \IF{$\delta$LML < tol \OR $\delta\theta$ < tol}
    \STATE Stop;
  \ENDIF
  \ENDFOR
  \end{algorithmic}
\end{algorithm}

A natural approach to implement the optimization in this context is the hybrid quantum-classical framework, where the gradient is estimated using the Hadamard test with sufficient measurements, and the gradient descent updates are performed classically, as shown in Algorithm \ref{alg:hybrid}. This approach can be interpreted as stochastic gradient descent (SGD) with an unbiased gradient estimator, and its general convergence properties are well-established \cite{harrow2021low}. Recent advancements have introduced a ``doubly stochastic'' gradient descent framework, which provides theoretical guarantees under the \textit{Polyak-Lojasiewicz} (PL) condition \cite{bottou2018optimization,sweke2020stochastic}. However, the Hadamard test requires $\mathcal{O}(1/\varepsilon^2)$ measurements to achieve a precision of $\varepsilon$, with its complexity further influenced by the gradient's variance. Moreover, measurements must be performed at each iteration (i.e. step (5) in Alg.~\ref{alg:hybrid}), adding to the computational burden.

\begin{algorithm}
  \caption{Quantum gradient descent algorithm for optimizing GP}\label{alg:qgd}
  \begin{algorithmic}[1]
  \REQUIRE Oracles $O_{\mathbf{y}}$, $O_{K}$, $O_{dK}$. The initial parameter $\theta_{0}$. 
  \STATE Construct $U_{\tilde{\mathbf{y}}_{l}}$, $U_{\tilde{\mathbf{y}}_{r}}$ as shown in Figure \ref{fig:circuit:initial};
  \STATE\label{alg:step:K-1} Construct $O_{K^{-1}}$ using QSVT and then $U$ as shown in Figure \ref{fig:circuit:U};
  \STATE Construct $U_{\psi_{0}}$ and then $G$ as shown in Figure \ref{fig:circuit:hadamard}, Figure \ref{fig:circuit:G};
  \FOR{$t = 0, 1, \dots, T-1$}
  \STATE\label{alg:step:qpe} Compute $\bra{\varphi_l} U \ket{\varphi_r}$ using QPE as shown in Figure \ref{fig:circuit:QPE}, Figure \ref{fig:circuit:aUb};
  \STATE Update $\ket{\theta_{t}}$ using classical arithmetic circuits as shown in Equation \eqref{eqn:quantum:GD};
  \STATE Reset the ancilla qubits;
  \STATE Estimate the LML as shown in \cite{zhao2019quantum};
  \IF{$\delta$LML < tol \OR $\delta\theta$ < tol}
    \STATE Stop;
  \ENDIF
  \ENDFOR
  \end{algorithmic}
\end{algorithm}

To address these limitations, we propose a quantum gradient descent algorithm that incorporates amplitude estimation to accelerate the Hadamard test and updates the hyperparameter as a quantum state, thereby avoiding intermediate measurements and reducing overall computational overhead.
We present the details of our algorithm as the following. 
 At the $t$-th step with $\ket{\theta_{t}}$, noting that $\bra{\varphi_1}U\ket{\varphi_2} = \Real\bra{\varphi_1}U\ket{\varphi_2}$ as it is a real number, we apply the amplitude estimation method described in Section \ref{sec:algorithm:aUb} to obtain
\begin{equation}
    \ket{\bra{\varphi_1}U\ket{\varphi_2}}  \ket{0^m} \ket{0^{2n+2a+b+4}} \ket{\theta_{t}}.
\end{equation}
Using classical arithmetic circuits again, for example the modular multiplication algorithm in \cite[Section 6.4.5]{rieffel2011quantum}, we can obtain
\begin{equation}\label{eqn:quantum:GD}
\begin{aligned}
    &\ket{\bra{\varphi_1}U\ket{\varphi_2}}  \ket{0^m} \ket{0^{2n+2a+b+4}} \ket{\theta_{t} - \frac{\eta_t \mathcal{N}_{l}\mathcal{N}_{r}}{2C_0^2} \bra{\varphi_l} U \ket{\varphi_r}} \\
    =& \ket{\bra{\varphi_1}U\ket{\varphi_2}}  \ket{0^m} \ket{0^{2n+2a+b+4}} \ket{\theta_{t+1}}.
\end{aligned}
\end{equation}
To reset the ancilla qubits, we can measure their states and apply the $X$ gate if the measurement result is $\ket{1}$. This procedure allows us to obtain the state $\ket{\theta_{t+1}}$ at the end.
The complete algorithm is provided in Algorithm \ref{alg:qgd} and a few remarks are listed in order. 
\begin{itemize}
\item First we emphrasize that $\ket{\theta_t}$ is represented in its $m$-qubit binary form, eliminating the need to measure $\ket{\theta_t}$ at each iteration of the procedure. The entire gradient descent process is fully embedded within the quantum circuit. 
\item Furthermore, during the iterative process, the value of the LML can be estimated \cite{zhao2019quantum}, enabling a line search procedure to determine the stepsize or 
formulating a stopping criterion as in Step (9) of Algorithm \ref{alg:qgd}. 
\item Finally we note that the number of measurements required in the gradient descent method remains independent of the variance of the LML or its gradient, which contrasts with the requirements for estimating the LML itself \cite{zhao2019quantum0, zhao2019quantum}. This property facilitates the efficient implementation of the optimization process without incurring additional overhead from repeated measurements.
\end{itemize}

\section{Runtime and error analysis}
\label{sec:analysis}
In this section, we provide runtime and error analysis of the proposed quantum gradient descent method.
To start we first analyze the error and runtime for a single iteration  of the algorithm, which is detailed in the following theorem:

\begin{mythm}[single step]
\label{thm:complexity:single}
Given the quantum state $\ket{\theta_{t}}$ encoding the binary representation of current parameter $\theta_{t}$ at time step $t$ for GP model selection, there exists a quantum algorithm using Oracles in Assumption \ref{assump:oracles} to prepare $\ket{\theta_{t+1}}$ with step size $\eta_t$ to accuracy $\varepsilon$, with a success probability at least $1-\delta$. Specifically, the algorithm uses $\mathcal{Q}_{0}$ ancilla qubits, $\mathcal{Q}_{1}$ queries of $U_{\varphi_l}, U_{\varphi_l}^\dagger, U_{\varphi_r}, U_{\varphi_r}^\dagger, O_{dK}, O_{dK}^\dagger$, $\mathcal{Q}_2$ queries of $O_{K}, O_{K}^\dagger$, controlled $O_{K}$, controlled $O_{K}^\dagger$, and $\mathcal{O}\left(\text{poly}(\mathcal{Q}_{0})\mathcal{Q}_1 + \text{poly}(n)\mathcal{Q}_{2}\right)$ additional $CNOT$ gates and single qubit gates, with
\begin{equation*}
    \mathcal{Q}_{0} = n+2a+b+4+ 2\left\lceil \log \frac{\mu_t}{\delta\sigma_N^4 \varepsilon} \right\rceil,
\end{equation*}
\begin{equation*}
    \mathcal{Q}_{1} = \mathcal{O}\left(\frac{\mu_t}{\delta\sigma_N^4\varepsilon}\right), \quad \mathcal{Q}_{2} = \mathcal{O} \left( \frac{\mu_t}{\delta \sigma_N^6 \varepsilon} \log \frac{\mu_t}{\sigma_N^8 \varepsilon} \right), \quad \mu_t = \eta_t \mathcal{N}_{l} \mathcal{N}_{r}.
\end{equation*}
\end{mythm}

\begin{proof}
There are two dominant sources of error: the inverse oracle $O_{K^{-1}}$ in Step \ref{alg:step:K-1} and the quantum phase estimation $U_{QPE}$ in Step \ref{alg:step:qpe}. To clarify the notations, we use $\tilde{O}_{K^{-1}}$ and 
\begin{equation}
    \tilde{U} = \left[ \tilde{O}_{K^{-1}} \right]_{w,a_2,\theta}  \left[ O_{dK} \right]_{w,b,\theta} \left[ \tilde{O}_{K^{-1}} \right]_{w, a_1,\theta}^{i_1},
\end{equation}
to denote the approximations of $O_{K^{-1}}$ and $U$ respectively. Given an $(s,a,0)$-block-encoding $O_{K}$ of $K$, a block-encoding $\tilde{O}_{K^{-1}}$ can be constructed via faster Hamiltonian simulation \cite{chakraborty2018power}, such that
\begin{equation}
    \left\| \left(I \otimes \bra{0^{a+1}} \right) \tilde{O}_{K^{-1}} \left( I \otimes \ket{0^{a+1}} \right) - C_0 K^{-1}(\theta) \right\| \leq \varepsilon_1.
\end{equation}
The total number of queries to $O_{K}$, $O_{K}^\dagger$, $CNOT$ gates and single-qubit gates is 
\begin{equation}\label{eqn:complexity:inverse}
    \mathcal{O} \left( \frac{sa}{C_0} \log^2 \left( \frac{1}{C_0^2 \varepsilon_1} \right) \right) = \mathcal{O} \left( \frac{1}{\sigma_N^2} \log^2 \left( \frac{1}{\sigma_N^4\varepsilon_1} \right) \right),
\end{equation}
where $s=1$ with Assumption \ref{assump:oracles}. Therefore 
\begin{equation*}
\begin{aligned}
    &\left\|  \left(I \otimes \bra{0^{a+1}}_{a_1} \bra{0^{a+1}}_{a_2} \bra{0^b}_{b}\right) \left(\tilde{U} - U \right) \left(I \otimes \ket{0^{a+1}}_{a_1} \ket{0^{a+1}}_{a_2} \ket{0^b}_{b}\right) \right\| \\
    \leq & \left\|  \left(I \otimes \bra{0^{a+1}}_{a_1}\right) \left( \left[ \tilde{O}_{K^{-1}} \right]_{w, a_1,\theta}^{i_1} - \left[ O_{K^{-1}} \right]_{w, a_1,\theta}^{i_1} \right) \left(I \otimes \ket{0^{a+1}}_{a_1}\right) \right\| \\
    &+ \left\|  \left(I \otimes \bra{0^{a+1}}_{a_2}\right) \left( \left[ \tilde{O}_{K^{-1}} \right]_{w, a_2,\theta} - \left[ O_{K^{-1}} \right]_{w, a_2,\theta} \right) \left(I \otimes \ket{0^{a+1}}_{a_2}\right) \right\| \\
    \leq & 2 \left\| \left(I \otimes \bra{0^{a+1}} \right) \tilde{O}_{K^{-1}} \left( I \otimes \ket{0^{a+1}} \right) - C_0 K^{-1}(\theta) \right\| \\
    \leq & 2 \varepsilon_1,
\end{aligned}
\end{equation*}
which implies
\begin{equation}
    \left| \bra{\varphi_l} \left( \tilde{U} - U \right) \ket{\varphi_r} \right| \leq 2 \varepsilon_1.
\end{equation}

Now let us consider the error introduced in the phase estimation approach in Section \ref{sec:algorithm:aUb}. In general, $\alpha/2\pi$ cannot be exactly represented by a $m$-bit number. In order to obtain the approximated phase $\tilde{\alpha}/2\pi$ to accuracy $\varepsilon_2 = 2^{-m}$ with a success probability at least $1-\delta$, we need $m + \lceil \log_2 \delta^{-1} \rceil$ ancilla qubits to store the value of the phase. Noting that $\Real \bra{\varphi_l} \tilde{U} \ket{\varphi_r} = -\cos \alpha$, the estimation error is bounded by
\begin{equation}
\begin{aligned}
    \left| \cos \alpha - \cos \tilde{\alpha} \right| = \left| -2 \sin \frac{\alpha + \tilde{\alpha}}{2} \sin \frac{\alpha - \tilde{\alpha}}{2} \right| \leq \left| \alpha - \tilde{\alpha} \right| \leq 2\pi \varepsilon_2.
\end{aligned}
\end{equation}
The procedure of QPE requires the implementation of QFT, inverse QFT and $G^{2^{j}}$ for $j=0,\cdots,m + \lceil \log_2 \delta^{-1} \rceil$, while $G=R_{\psi_0} R_{\text{good}}$ consists of $\mathcal{O}(n+a+b)$ single Gates and 2-qubit Gates, 2 queries of $U_{\varphi_l}$, $U_{\varphi_r}$, $U$ and their conjugate transpose. Therefore the query complexity of $U_{\varphi_l}$, $U_{\varphi_r}$, $\tilde{U}$ is 
\begin{equation}\label{eqn:complexity:qpe}
    \mathcal{O} \left( \sum_{j=0}^{m + \lceil \log_2 \delta^{-1} \rceil} 2^j \right) = \mathcal{O} \left( 2^{m + \lceil \log_2 \delta^{-1} \rceil} \right) = \mathcal{O} \left( \frac{1}{ \delta \varepsilon_2} \right),
\end{equation}
where $\tilde{U}$ is obtained with $\tilde{O}_{K^{-1}}$ (constructed using $O_{K}$) and $O_{dK}$.

In all, the error of approximating $\bra{\varphi_l} U \ket{\varphi_r}$ is
\begin{equation*}
\begin{aligned}
    \left| - \cos \tilde{\alpha} - \bra{\varphi_l} U \ket{\varphi_r} \right| &\leq \left| -\cos \tilde{\alpha} - (- \cos \alpha) \right| + \left| -\cos \alpha - \bra{\varphi_l} U \ket{\varphi_r} \right| \\
    &\leq 2\pi \varepsilon_2 + \left| \bra{\varphi_l} \tilde{U} \ket{\varphi_r} - \bra{\varphi_l} U \ket{\varphi_r} \right| \\
    &\leq 2\pi \varepsilon_2 + 2 \varepsilon_1.
\end{aligned}
\end{equation*}
To achieve the tolerance
\begin{equation}
   \left| \tilde{\theta}_{t+1} - \theta_{t+1} \right| = \frac{\mu_t}{2C_0^2} \left| - \cos \tilde{\alpha} - \bra{\varphi_l} U \ket{\varphi_r} \right| \leq \varepsilon,
\end{equation}
we choose $\varepsilon_1 = \frac{C_0^2\varepsilon}{2\mu_t}$, $\varepsilon_2 = \frac{C_0^2\varepsilon}{2\mu_t \pi}$. Substituting $\varepsilon_1, \varepsilon_2$ into the query complexity results \eqref{eqn:complexity:inverse} and \eqref{eqn:complexity:qpe} with $C_0=\mathcal{O}(\sigma_N^2)$, the number of queries of $U_{\varphi_l}, U_{\varphi_l}^\dagger, U_{\varphi_r}, U_{\varphi_r}^\dagger, O_{dK}, O_{dK}^\dagger$, $O_{K}, O_{K}^\dagger$, controlled $O_{K}$, and controlled $O_{K}^\dagger$ is obtained.

\end{proof}

Building upon Theorem \ref{thm:complexity:single}, 
we can now obtain the overall complexity of the method across multiple iteration steps, 
as stated in the following theorem. 

\begin{mythm}[multiple steps]
\label{thm:complexity:multi}
Let the task be: to use  the quantum gradient descent method for $T \geq 0$ steps to prepare a solution $\ket{\theta_{T}}$ to final accuracy $\varepsilon^\prime$, with step-size schedule $\{\eta_t\}_{t=0}^{T-1} >0$. There exists a quantum algorithm for multiple gradient steps with a success probability at least $(1-\delta)^T$, that uses $\mathcal{Q}_{1}^\prime$ queries of $U_{\varphi_l}, U_{\varphi_l}^\dagger, U_{\varphi_r}, U_{\varphi_r}^\dagger, O_{dK}, O_{dK}^\dagger$, $\mathcal{Q}_2^\prime$ queries of $O_{K}, O_{K}^\dagger$, controlled $O_{K}$, controlled $O_{K}^\dagger$, and $\mathcal{O}\left(\text{poly}(\mathcal{Q}_{0})(\mathcal{Q}_1^\prime + \text{poly}(n)\mathcal{Q}_{2}^\prime)\right)$ additional $CNOT$ gates and single qubit gates, with
\begin{equation*}
    \mathcal{Q}_{1}^\prime = \mathcal{O}\left(\frac{TC_2\mu}{(1+C_2)^T\delta\sigma_N^4\varepsilon^\prime}\right), \quad \mathcal{Q}_{2}^\prime = \mathcal{O} \left( \frac{TC_2\mu}{(1+C_2)^T\delta \sigma_N^6 \varepsilon^\prime} \log \frac{C_2\mu}{(1+C_2)^T\sigma_N^8 \varepsilon^\prime} \right), \quad \mu = \max_{0\leq t < T} \mu_t.
\end{equation*}
\end{mythm}

\begin{proof}
    For simplicity, we denote $\tilde{\theta}_{t}$ the approximation of $\theta_{t}$ with error $\varepsilon_t = \left|\tilde{\theta_{t}} - \theta_{t}\right|$ and denote $\theta_{t+1} = \theta_{t} - g(\theta_{t})$ the gradient descent method \eqref{eqn:GD}. The error of $\tilde{\theta_{t+1}}$ can be calculated as
    \begin{equation*}
    \begin{aligned}
        \left| \tilde{\theta}_{t+1} - \theta_{t+1} \right| &\leq \left| \tilde{\theta}_{t+1} - (\tilde{\theta}_{t} - g(\tilde{\theta}_{t})) \right| + \left| (\tilde{\theta}_{t} - g(\tilde{\theta}_{t})) - \theta_{t+1} \right| \\
        &\leq \varepsilon + \left| \tilde{\theta}_{t} - \theta_{t} \right| + \left| g(\tilde{\theta}_{t}) - g(\theta_{t}) \right| \\
        &\leq \varepsilon + (1 + C_{2}) \varepsilon_{t},
    \end{aligned}        
    \end{equation*}
    where $C_{2} = \max_{\theta,t} \left| \eta_t \frac{\partial^2}{\partial\theta^2}\log p(\mathbf{y}|X,\mathbf{\theta}) \right|$ is a constant. By induction, it is easy to calculate 
    \begin{equation}
        \varepsilon_{t+1} \leq (1+C_2)^{t+1} \varepsilon_{0} + \frac{(1+C_2)^{t+1} - 1} {C_2} \varepsilon \leq \frac{(1+C_2)^{t+1}} {C_2} \varepsilon.
    \end{equation}
    To achieve accuracy $\varepsilon^\prime$ at step $T$, we need to choose $\varepsilon = C_2 \varepsilon^\prime / (1+C_2)^{T+1}$. Substituting $\varepsilon$ into the query complexity result in Theorem \ref{thm:complexity:single} and summing over the query complexity of all steps, the desired result follows directly.
\end{proof}

As stated in Theorem \ref{thm:complexity:multi}, the query complexity of the oracles depends on $T$, the number of iterations in the gradient descent method. We have the following remark regarding parameter $T$:
\begin{myremark}
      For a convex objective, classical gradient descent achieves convergence within $\mathcal{O}(1/\varepsilon^\prime)$ iterations, while for a more general objective, the convergence rate is slower, typically of a lower order. In practice, it is often preferable to fix the number of iterations, yielding a near-optimal hyperparameter. Regardless of the scenario, we emphrasize that $T$ does not explicitly depend on the size of the dataset, $N = 2^n$. 
\end{myremark}

In Theorem \ref{thm:complexity:multi} we analyze the query complexity of the proposed algorithm.
Next we shall consider the total time cost  by taking into account
of the cost for oracle access to $O_{K}$, $O_{dK}$.
The analysis is done under the assumption of sparse access oracles and within the framework of block-encoding. 
Specifically, given sparse access to an $s$-sparse matrix $A\in N\times N$ with $N = 2^n$, it is possible to construct an $(s, \log_2(s)+1)$-block-encoding of $A/\left\|A\right\|_{\max}$ \cite{camps2022explicit}. Therefore we have the following corollary:
\begin{mycoro}\label{coro:complexity:total}
    Given sparse accesses to the $s$-sparse kernel matrix $K=K_f + \sigma_N^2 \mathbf{I}$ and its derivative $\frac{dK}{d\theta}$, the $(s, \log_2(s)+1)$-block-encodings $O_K$ and $O_{dK}$ can be constructed. The proposed quantum gradient descent method for $T \geq 0$ steps prepares a solution $\ket{\theta_{T}}$ to final accuracy $\varepsilon^\prime$ in cost
    \begin{equation}
        \tilde{\mathcal{O}} \left( sT\mu\polylog(N)/\delta\sigma_N^6\varepsilon^\prime \right),
    \end{equation}
    with some logarithmic terms  ignored.
\end{mycoro}
\begin{proof}
    The corollary is a direct consequence of  
    \cite[Theorem 4.1]{camps2022explicit}, Eq.~\eqref{eqn:complexity:inverse} and Theorem \ref{thm:complexity:multi}.
\end{proof}

\begin{myremark}
    The complexity of block-encoding is determined by the sparsity $s$ of the kernel matrix $K$. When $K$ is sparse and well-conditioned, the proposed quantum algorithm achieves an exponential speed-up compared to classical algorithms. Even in the worst-case scenario, where $K$ is a dense matrix with sparsity $s = N$, the quantum algorithm still offers a cubic speed-up over its classical counterparts.
\end{myremark}

It is also worth mentioning that recent advancements in encoding techniques, such as those introduced in \cite{nguyen2022block,wossnig2018quantum}, enable efficient encoding of dense kernel matrices,
reducing its complexity to $\polylog(N)\kappa$.  
These coding methods can be incorporated with the proposed algorithm,
providing efficient quantum computation for problems with dense kernel matrices.  


\section{Conclusion}
\label{sec:conclusion}
While GPR is a powerful tool for supervised machine learning, 
its classical  implementations face significant computational challenges,
particularly with the $\mathcal{O}(N^3)$ complexity of kernel matrix inversion, limiting scalability to large datasets.
This issue is especially serious when the GP model requires optimization over certain hyperparameters, as the $\mathcal{O}(N^3)$ 
matrix operations must be performed during each iteration of the optimization process.
Quantum computing offers a transformative solution, with quantum algorithms like the HHL method providing exponential speedups for certain key operations. This work presents a quantum algorithm for conducting gradient descent optimization of the GP model over hyperparameters.
A notable feature of the proposed algorithm is that, by encoding the hyperparameter updates $\ket{\theta_t}$ directly into the quantum circuit, the need for repeated measurements at each iteration is eliminated, significantly simplifying the gradient descent procedure. Furthermore, the number of measurements required at the end of the algorithm remains independent of the variance of the log marginal likelihood function or its gradient, ensuring a computationally efficient implementation.

Runtime and error analysis is also provided for the proposed algorithm. We show that, the query complexity of the proposed quantum algorithm depends on the number of gradient descent iterations, the number of hyperparameters, and the specified precision, none of which explicitly depend on the dataset size $N=2^n$. Consequently, the overall query complexity scales polynomially with $n$ and polylogarithmically with  $N$. We also consider the total time cost, which relies on the block-encoding technique employed for the kernel matrices. 
We demonstrate that under the assumption of sparse access oracles
the proposed method can offer at least a cubic speed-up, and, in certain cases, exponential speed-up. 
These findings underscore the substantial computational advantages of quantum algorithms in addressing data intensive problems, enabling broader applications of Gaussian Process Regression. 

Finally we want to discuss some limitations and challenges of the proposed method, which we hope to address in the future. 
Certain kernel structures can result in complex oracle designs, leading to reduced efficiency. Addressing this limitation may require developing more generalized and robust oracle designs. Moreover, the algorithm naturally updates multiple hyperparameters, by computing partial derivatives for each one. This can lead to increased computational costs when the dimensionality of the hyperparameters is high.
We expect that many advanced optimization techniques developed in the classical setting,
such as the randomized coordinate descent method \cite{lu2015complexity}, may be employed to further reduce computational complexity.
These issues will be  investigated in future work.

\section*{Ackownledgement}

SJ and LZ were partially supported by the NSFC grant No. 12031013, the Shanghai Jiao Tong University 2030 Initiative, the Shanghai Municipal Science and Technology Major Project (24LZ1401200), and the Fundamental Research Funds for the Central Universities.

\bibliographystyle{plain}
\bibliography{ref}

\appendix

\section{Construction of $U_{\tilde{\mathbf{y}}_{l}}$ and $U_{\tilde{\mathbf{y}}_{r}}$}
\label{sec:appendix:initial}

For simplicity, we consider the following state
\begin{equation}\label{eqn:state:initial}
\begin{aligned}
    \ket{\tilde{\mathbf{y}}_{\gamma}} &= \sum_{j\in[2^n]} \frac{\cos\gamma}{\sqrt{2^n}} \ket{0}_{i_1} \ket{j}_{i_2} \ket{j}_{w} + \sin \gamma \ket{1}_{i_1}\ket{0^n}_{i_2}\ket{\mathbf{y}}_{w}.
\end{aligned}
\end{equation}
When applied to $\tilde{\mathbf{y}}_{l}$ and $\tilde{\mathbf{y}}_{r}$, we can take $\gamma = \arctan \frac{\left\| \mathbf{y} \right\|}{\sqrt{2^n}}$ and $\gamma = \pi - \arctan \frac{\left\| \mathbf{y} \right\|}{\sqrt{2^n} C_0}$ respectively.

Starting from the initial state $\ket{0}_{i_1}\ket{0^n}_{i_2}\ket{0^n}_{w}$, $RY(2\gamma)$ is applied to the index register $i_1$:
\begin{equation}
    RY(2\gamma) \ket{0}_{i_1} = \cos\gamma \ket{0}_{i_1} + \sin\gamma \ket{1}_{i_1}.
\end{equation}
To construct the first term in Equation \eqref{eqn:state:initial}, the Hadamard gate $H^{\otimes n}$ is applied to the index register $i_{2}$, controlled by the index register $i_{1}$ being in state $\ket{0}$, 
\begin{equation}
    \left[ H^{\otimes n} \right]^{i_1}_{i_2} ( \cos\gamma \ket{0}_{i_1} + \sin\gamma \ket{1}_{i_1}) \ket{0^n}_{i_2}\ket{0^n}_{w} = \sum_{j\in[2^n]} \frac{\cos\gamma}{\sqrt{2^n}} \ket{0}_{i_1} \ket{j}_{i_2} \ket{0^n}_{w} + \sin \gamma \ket{1}_{i_1}\ket{0^n}_{i_2}\ket{0^n}_{w}.
\end{equation}
Then CNOT gates are applied to the index register $i_{2}$ and the working register $w$, which turns the state into
\begin{equation}
    \sum_{j\in[2^n]} \frac{\cos\gamma}{\sqrt{2^n}} \ket{0}_{i_1} \ket{j}_{i_2} \ket{j}_{w} + \sin \gamma \ket{1}_{i_1}\ket{0^n}_{i_2}\ket{0^n}_{w}.
\end{equation}
Finally, the oracle $O_{\mathbf{y}}$ is applied to the working register $w$, controlled by the index register $i_1$ being in state $\ket{1}$, leading to
\begin{equation}
    \sum_{j\in[2^n]} \frac{\cos\gamma}{\sqrt{2^n}} \ket{0}_{i_1} \ket{j}_{i_2} \ket{j}_{w} + \sin \gamma \ket{1}_{i_1}\ket{0^n}_{i_2}\ket{\mathbf{y}}_{w}.
\end{equation}
The entire procedure is illustrated in Figure \ref{fig:circuit:initial}.

\begin{figure}[htbp]
    \centering
    \begin{quantikz}
        \lstick{$\ket{0}_{i_1}$} &  & \gate{RY(2\gamma)} &  \gate{X} & \ctrl{1} & \cdots & \ctrl{2} & \gate{X} & & \cdots  &  &  \ctrl{4} &  \\
        \lstick[2]{$\ket{0^n}_{i_2}$} & \qw{\vdots} &  &  & \gate{H} & \cdots &  &  &  \ctrl{2} & \cdots  &  &  &  \\
        &  &  &  &  & \cdots & \gate{H} &  &  & \cdots & \ctrl{2} &  &  \\
        \lstick[2]{$\ket{0^{n}}_{w}$} & \qw{\vdots}  &  &  &  &  \cdots &  &  & \targ{} &  \cdots  &  &  &   \\
        & & & & & & & & &  \cdots & \targ{} &  \gate{O_{\mathbf{y}}} & 
    \end{quantikz}
    \caption{Quantum circuit for initial state preparation}.
    \label{fig:circuit:initial}
\end{figure}
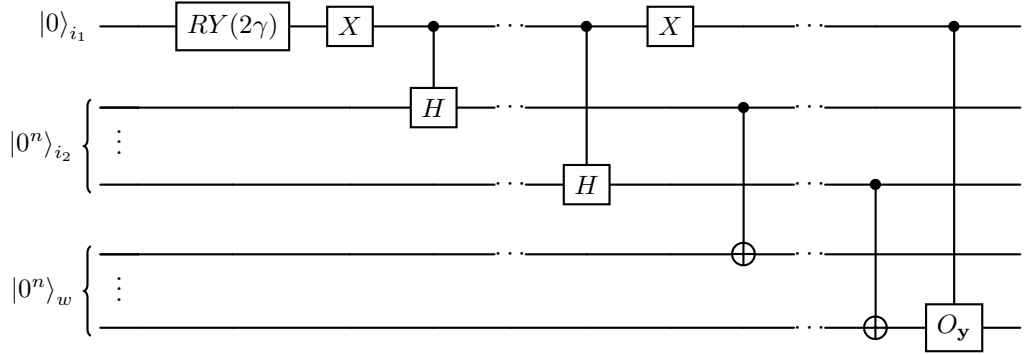

\end{document}